\documentclass[11pt, a4paper]{article}

\usepackage[a4paper, total={16cm, 25cm}]{geometry}
\usepackage{amsmath}
\usepackage{amssymb}
\usepackage{graphicx}
\usepackage{amsthm}
\usepackage{color}
\usepackage{footmisc}
\usepackage{algorithm}
\usepackage{algorithmic}

\newtheorem{theorem}{Theorem}

\newtheorem{rational for conjecture}{Rational for Conjecture}

\newcommand{\INPUT}{\textbf{Input: }}

\newcommand{\METHOD}{\textbf{Method: }}

\begin{document}

\LARGE

\begin{center}
\noindent \textbf{Parameter Estimation in Abruptly Changing Dynamic Environments}  
\end{center}

\Large

\begin{center}
  Hugo Lewi Hammer and Anis Yazidi  
\end{center}
\large
\begin{center}
  OsloMet - Oslo Metropolitan University\\
  Norway
\end{center}

\normalsize

\begin{abstract}
Many real-life dynamical systems change abruptly followed by almost stationary periods. In this paper, we consider streams of data with such abrupt behavior and investigate the problem of tracking their statistical properties in an online manner.

We devise a tracking procedure where an estimator that is suitable for a stationary environment is combined together with an event detection method such that the estimator rapidly can jump to a more suitable value if an event is detected. Combining an estimation procedure with detection procedure is commonly known idea in the literature. However, our contribution lies in building the detection procedure based on the difference between the stationary estimator and a Stochastic  Learning  Weak  Estimator (SLWE). The SLWE estimator is known to be the state-of-the art approach to tracking properties of non-stationary environments and thus should be a better choice to detect changes in abruptly changing environments than the far more common sliding window based approaches. To the best of our knowledge, the event detection procedure suggested by Ross et al. (2012) \cite{ross2012} is the only procedure in the literature taking advantage of the powerful tracking properties of the SLWE estimator. The procedure in \cite{ross2012} is however quite complex and not well founded theoretically compared to the procedures in this paper. In this paper, we focus on estimation procedure for the binomial and multinomial distributions, but our approach can be easily generalized to cover other distributions as well.

Extensive simulation results based on both synthetic and real-life data related to news classification demonstrate that our estimation procedure is easy to tune and performs well.
\end{abstract}

\section{Introduction}


The Maximum Likelihood Estimates (MLE) as well as the Bayesian
estimation families operate with the premise that the distribution of the data being estimated is stationary over time.
Under such settings, the convergence to the true value of the parameter being estimated takes place with probability 1 when
the number of samples tends to infinity.
However, in many real-life applications, the assumption on the stationarity of the data does not hold and the true underlying parameter being estimated changes over time.
In this paper, we consider the problem of estimating  binomial and multinomial random variables which vary over
time. The Stochastic  Learning  Weak  Estimators  (SLWE) are known to be the state-of-the-art approach for such an estimation problem \cite{OommenRueda06, ZhangOommen11}. The SLWE enjoys a multiplicative update form that makes it superior to the state-of-the-art estimation approaches which are mainly of additive flavor.
However, the right choice of the intrinsic parameter of the SLWE, $\lambda$, is still an open issue. The latter parameter controls the forgetting of old data and controls the ability of the scheme to adapt to changes in the environments. If the system changes rapidly the parameter should be chosen to rapidly forget the old stale data. On the other hand, if the environment is stabilizing, the rate of forgetting should decrease.

The SLWE has found numerous successful applications in the literature. Applications of the SLWE include adaptive classifiers for spam filtering \cite{ZhangOommen11}, adaptive file encoding with nonstationary distributions \cite{Rueda06}, intrusion detection in computer networks \cite{Tart006}, tracking shifts of languages in online discussions \cite{Stensby13},  learning user preferences under concept-shift \cite{OommenAdapt2012, YazidiPref2011}, fault-tolerant routing in Ad-hoc networks \cite{Oommen09}, digital content forensics for detecting illicit images \cite{ibrahim2009detecting}, detection and tracking of malicious nodes in both Ad-hoc networks \cite{rikli2016lightweight}, vehicular mobile WiMAX networks \cite{misra2015selfishness}, and optimizing firewall matching time via dynamic rule ordering \cite{mohan2016dynamic}-- to mention a few.

In many of such practical problems the dynamical system changes abruptly followed by periods where the system is almost stationary. Unfortunately, the SLWE is not well suited for such cases. By choosing a low value of $\lambda$, the estimator will rapidly adjust after an abrupt change, but on the other hand, it will result in a higher estimation uncertainty when the system stabilizes. By choosing a high value of $\lambda$, the estimation uncertainty will be low in the stationary parts, but on the other hand, the estimation procedure will suffer from adjusting too slowly after an abrupt change.

In this paper, we suggest a computationally efficient estimation procedure for a dynamical system that contains both abrupt changes and stationary parts.
The estimator combines an estimator that is suitable for the stationary parts together with an event detection procedure. When an abrupt change is detected, the estimator rapidly jumps to a more suitable estimate. The far most common event detection approach is to compare the properties of the data stream on a long term time window with a more short term time window \cite{gama2013}. In such window based approaches, each sample in the window is given an equal weight, but intuitively it is more reasonable to give more weight to the most recent data which is done by the SLWE. In this paper we therefore suggest to build the event detection procedure by comparing the estimate by the stationary estimator with an SLWE estimator. Through lightweight and subtle hypothesis testing mechanisms, we decide, in each iteration, if the stationary estimate should jump to new value (event detected) or not. Quite surprisingly, we have found only one other paper in the literature using the advantages of the SLWE for event detection, namely the paper by Ross et al. (2012) \cite{ross2012}. Compared to \cite{ross2012}, our suggested approach is simpler and better founded theoretically. We present the estimation procedure for the binomial and multinomial distributions, but can be applied to other distributions as well.
The article is organized as follows.
In Section \ref{sec:RelatedWork} we review related work. In Section \ref{sec:weak} we present the SLWE estimator for a stream of Bernoulli variables and in Section \ref{sec:shiftenv} we present the details of our approach. In Section \ref{sec:multinom} we extend the scheme to the multinomial case. Finally, in Section \ref{sec:exp} we perform thorough evaluation of the algorithms and draw some conclusions in Section \ref{sec:closrem}.

\section{Related Work and State-of-the-Art}
\label{sec:RelatedWork}
In this Section we review related work.
First, in Section \ref{sec:estimation} we will review legacy scheme for estimation under non-stationary environment.
Then, in Section \ref{sec:parameters} we will review the different approaches for controlling the parameters of estimators operating in non-stationary environments.

\subsection{Estimation in Non-Stationary Environments}
\label{sec:estimation}


Probably, the most classical and utilized method for dealing with non-stationary estimation problems is the
\emph{sliding window} approach which can be seen as a short memory version of the MLE.
According to the sliding window approach, online the last samples that fit in the window are used to compute the estimates.
Nevertheless, the \emph{sliding window} method suffers from a tuning problem. In fact, if the size of the window is chosen too large, then the quality of the estimates will be deteriorated by stale data values, while choosing a too small window size would rather lead to poor estimates with low confidence.

A myriad of works have been proposed to address detecting change points.
Those methods fall under two main families: Page's cumulative sum
(CUSUM) \cite{Bass93} detection procedure, and the
Shiryaev-Roberts-Pollak detection procedure. In \cite{Shir78},  Shiryayev resorted to a Bayesian formulation  in which the change point is assumed to have a geometric
prior distribution. CUSUM uses the idea of maximum
likelihood ratio test hypothesis  to discern change points.
However, a downside of these two approaches is their computational complexity which renders the SLWE as well as the estimator in this paper lightweight alternatives.


When it comes to extensions of the sliding window, Koychev \emph{et al}. proposed a new paradigm called Gradual Forgetting (GF) \cite{Koychev00b,Koychev06,Koychev00a}.
According to the principles of GF, observations in the same window are treated unequally when computing the estimates based on weight assignment.
Recent observations receive more weights than distant ones. Different forgetting functions were proposed
ranging from linear \cite{koychev2000} to exponential \cite{klinkenberg2004}.


In \cite{OommenRueda06}, Oommen and Rueda presented the SLWE to estimate the underlying parameters of  time varying binomial/multinomial distribution.
The SLWE  originally stems from the theory of variable structure Learning Automata \cite{narendra2012learning}, and more particularly, its reward-inaction flavor.
The most appealing properties of the SLWE which makes it the state-of-the-art is its multiplicative form of updates.
Two different counter-parts of SLWE \cite{OommenRueda06} for discretized spaces was recently proposed in \cite{yazidi2016stochastic} and \cite{YazidiO16}. In a similar manner to the SLWE, the latter solution also suffers from the problem of tuning the resolution parameter.

%



\subsection{Estimation using Adjustable parameters}
\label{sec:parameters}


In this Section, we survey some of the most pertinent techniques for estimation in dynamic environments that are orthogonal to the SLWE.
For a thorough survey we refer the reader to the surveys \cite{gama2013,kulkarni2014} which provide a comprehensive taxonomy of estimation methods in non-stationary environments, namely, \emph{adaptive windowing},  \emph{aging factors}, \emph{instance selection} and \emph{instance weighting}.

Gama  \emph{et al}.  \cite{gama2013} presents a clear distinction between memory management and forgetting mechanisms.
Adaptive windowing \cite{widmer1996} works with the premise of growing the size the sliding window indefinitely until a change is detected via a change detection technique. In this situation, the size of the window is reset whenever a changed is detected.

Another interesting family of approaches assume that the true value of the parameter being estimated is revealed after some delay, which
enables quantifying the error of the estimator. In such settings, some research \cite{tsymbal2008} have used ensemble methods where the output of different estimators is combined using weighted majority voting. The weights of each estimator is adjusted based on its error. In this sense, estimation methods that produce high error see their weight decrease.

In the same perspective,  the estimated error can be used for re-initializing the estimation as performed in \cite{ross2012}.
In all brevity, changes are detected based on comparing sections of data, using statistical analysis to detect distributional changes, i.e., abrupt or gradual changes in the mean of the data points when compared with a baseline mean with a random noise component. One option is also to keep a reference window and compare recent windows with the reference window to detect changes \cite{dries2009}. This can, for example, be done based on comparing the probability distributions of the reference window and the recent window using Kullback-Leibler divergence \cite{dasu2006,sebastiao2007}.


\section{Stochastic Learning Weak Estimator}
\label{sec:weak}

Let $X_1, X_2, X_3, \ldots$ represent a stream of independent and identically distributed Bernoulli stochastic variables with parameter $p$. That is
\begin{align}
  \label{eq:1}
  \begin{split}
  P(X_n = 0) &= 1-p\\
  P(X_n = 1) &= p
  \end{split}
\end{align}
for $n=1,2,3,\ldots$.

We now want to estimate the parameter $p$ from the stream of Bernoulli variables. Using the weak estimator, the estimate of $p$ is updated by the following recursion
\begin{align}
  \label{eq:2}
  \begin{split}
  \hat{p}_1 &= X_1\\
  \hat{p}_n &= \lambda_n \hat{p}_{n-1} \hspace{20mm}\text{ if } X_n = 0 \\
  \hat{p}_n &= 1 - \lambda_n (1 - \hat{p}_{n-1}) \,\,\,\text{ if } X_n = 1
  \end{split}
\end{align}
where $\hat{p}_n$ represents the estimate of $p$ after the arrival of $X_n$ and $\lambda_n,\,\, n = 1,2,\ldots$ are constants between zero and one. The intuition is that if $X_n = 0$ we should reduce our current estimate of $p$ (the probability of one) which is achieved by multiplying the current estimate of $p$ by $\lambda_n$. On the other hand, if $X_n = 1$ we should reduce the estimate of $1-p$ (the probability of zero) which gives
\begin{align*}
  1 - \hat{p}_n &= \lambda_n (1 - \hat{p}_{n-1}) \\
  \hat{p}_n &= 1 - \lambda_n (1 - \hat{p}_{n-1}) \\
\end{align*}
which is equal to the last equation in \eqref{eq:2}.

The recursions in \eqref{eq:2} can be written as follows
\begin{align*}
  \hat{p}_n = X_n(1 - \lambda_n (1 - \hat{p}_{n-1})) + (1-X_n) \lambda_n \hat{p}_{n-1}, \,\, n=1,2,\ldots
\end{align*}
with $\lambda_1 = 0$. Using straight forward calculations this simplifies to
\begin{align}
  \label{eq:3}
  \hat{p}_n = \lambda_n \hat{p}_{n-1} + (1 - \lambda_n)X_n
\end{align}
which can be recognized as the exponentially weighted moving average.

We can prove by induction that $\hat{p}_n$ is an unbiased estimator for $p$ for every $n$ as follows
\begin{align*}
  E(\hat{p}_1) &= E(X_1) = p\,\,\,\text{ (recall } \lambda_1 \text{ is set to } 0 \text{)}\\
  E(\hat{p}_n) &= E(\lambda_n \hat{p}_{n-1} + (1 - \lambda_n)X_n) \\
                      &= \lambda_n p + (1-\lambda_n) p\\
                       &= p
\end{align*}
The variance depends on the choice of the  $\lambda$'s. We look at two special cases.\\
\textit{$\lambda$ constant}: It can be proved that if we set all $\lambda_n = \lambda$, the limiting variance is given by \cite{ZhangOommen11}
\begin{align*}
  \lim_{n \rightarrow \infty} \text{Var}\left(\hat{p}_n\right) = \frac{1-\lambda}{1+\lambda}p(1-p)
\end{align*}
An advantage of the constant $\lambda$ approach is that if the value of $p$ is changing with time in the underlying Bernoulli data stream, the estimator will rapidly adjust to these changes \cite{ZhangOommen11}. A disadvantage is that if $p$ is not changing, the variance of the estimator will have a lower limit and never reaches zero.

\textit{Sample mean}: The sample mean is the maximum likelihood estimator of $p$ and is the natural estimator to use if $p$ is not changing with time. Let $\bar{p}_{n-1}$ denote the sample mean of the first $n-1$ Bernoulli variables from the stream
\begin{align*}
  \bar{p}_{n-1} = \frac{1}{n-1} \sum_{i=1}^{n-1} X_i
\end{align*}
When $X_n$ arrives, the sample mean can be updated as follows
\begin{align}
  \label{eq:4}
  \bar{p}_{n} = \frac{n-1}{n}\bar{p}_{n-1} +  \frac{1}{n} X_n
\end{align}
which is equivalent to \eqref{eq:3} with $\lambda_n = (n-1)/n$. This means that the sample mean is a special case of the general recursion in \eqref{eq:3}. It is well known that $\lim_{n \rightarrow \infty} \text{Var}\left(\bar{p}_n \right) = 0$. A disadvantage of the sample mean is that if $p$ is changing with time, the sample mean will become very slow at adjusting to these changes. On the other hand if $p$ is not changing, the sample mean is the optimal estimator in the sense that no other unbiased estimators can achieve less variance.

\section{Estimation in a shifting environment}
\label{sec:shiftenv}

Suppose a situation where $p$ is switching between different values with time. An example could be a news stream where the topic of the news stream suddenly changes due to different real life events. Another example could be a machine operated by different employees at different time periods each with its own error rate characterized by $p$. We assume that the instants in which $p$ switches value are unknown.

For such systems a natural strategy would be to use the sample mean whenever the $p$ is not changing, and a mechanism to ``jump'' fast towards a new estimate if the value of $p$ has changed. In this paper we suggest a method that combines the sample mean and a weak estimator with constant $\lambda$. Ross et al. (2012) \cite{ross2012} is the only paper we have found in the literature that uses the same idea. Let $\hat{p}^{\lambda}_n$ and $\bar{p}_n$ denote the weak estimator with constant $\lambda$ and the sample mean, respectively, after the arrival of $X_n$. If $p$ switches value, $\hat{p}^{\lambda}_n$ will rapidly adjust to the new value of $p$, while minor changes will appear to $\bar{p}_n$. This can be used to build an efficient method to detect changes in $p$ and ``jump'' to the new value of $p$. The key ingredient will be the distribution of the difference between the estimators $\hat{p}^{\lambda}_n - \bar{p}_n$. If $p$ switches value we expect that $\hat{p}^{\lambda}_n - \bar{p}_n$ will be large in absolute value and larger then what we would expect if $p$ remains constant. This can be used to build a statistical test if $p$ has changed value or not. We start by presenting the expectation and variance of this distribution.
\begin{theorem}
\label{the:1}
  Let $X_1, X_2, X_3, \ldots$ represent a stream of independent and identically distributed Bernoulli stochastic variables with parameter $p$. Further let $\hat{p}^{\lambda}_n$ and $\bar{p}_n$ denote the weak estimator with constant $\lambda$ and the sample mean, respectively, after the arrival of $X_n$. Then
  \begin{align}
    &\text{E}\left(\hat{p}^{\lambda}_n - \bar{p}_n \right) = 0 \\
    \begin{split}
      \label{eq:5}
    &\text{Var}\left(\hat{p}^{\lambda}_n - \bar{p}_n \right) = p(1-p) \left( \left( \frac{1}{n} - \lambda^{n-1} \right)^2 + \sum_{i=2}^n \left( \frac{1}{n} - (1-\lambda) \lambda^{n-i} \right)^2\right)
    \end{split}
  \end{align}
\end{theorem}
\begin{proof}
We start by computing through the recursions in \eqref{eq:3}
\begin{align*}
  \hat{p}_1 &= X_1\\
  \hat{p}_2 &= \lambda_2 \hat{p_1} + (1-\lambda_2)X_2\\
            &= \lambda_2 X_1 + (1-\lambda_2)X_2\\
  \hat{p}_3 &= \lambda_3\hat{p_2} + (1-\lambda_3)X_3\\
            &= \lambda_3[\lambda_2 X_1 + (1-\lambda_2)X_2] + (1-\lambda_3)X_3\\
            &= \lambda_3 \lambda_2 X_1 + \lambda_3 (1-\lambda_2)X_2 + (1-\lambda_3)X_3\\
            &\vdots \hspace{2cm} \vdots\\
  \hat{p}_n &= \sum_{i=1}^n X_i(1-\lambda_i)\prod_{j=i+1}^n \lambda_j, \text{ with } \lambda_1 = 0
\end{align*}
Setting $\lambda_n = \lambda$ (and still $\lambda_1 = 0$) we get
\begin{align*}
  \hat{p}_n = X_1 \lambda^{n-1} + \sum_{i=2}^n X_i (1-\lambda) \lambda^{n-i}
\end{align*}
and setting $\lambda_n = (n-1)/n$ we get the sample mean $\hat{p}_n = \frac{1}{n} \sum_{i=1}^n X_i$.

Now we are ready to compute the expectation and variance
\begin{align*}
  E(\hat{p}^{\lambda}_n - \bar{p}_n)& = E\left( \frac{1}{n} \sum_{i=1}^n X_i - X_1 \lambda^{n-1} - \sum_{i=2}^n X_i (1-\lambda) \lambda^{n-i} \right) \\
                            &= \frac{1}{n} \sum_{i=1}^n p - p \lambda^{n-1} - p\sum_{i=2}^n (1-\lambda) \lambda^{n-i} = \\
                            &= p - p \lambda^{n-1} - p(1-\lambda) \sum_{i=0}^{n-2} \lambda^i\\
                            &= p - p \lambda^{n-1} - p(1-\lambda) \frac{1 - \lambda^{n-1}}{1-\lambda} \\
                            &= 0
\end{align*}
\begin{align*}
  \text{Var}(\hat{p}^{\lambda}_n - \bar{p}_n) &= \text{Var} \left( \frac{1}{n} \sum_{i=1}^n X_i - X_1 \lambda^{n-1} - \sum_{i=2}^n X_i (1-\lambda) \lambda^{n-i} \right) = \\
  & = \text{Var}\left( \left( \frac{1}{n} - \lambda^{n-1} \right)X_1 + \sum_{i=2}^n \left( \frac{1}{n} - (1-\lambda) \lambda^{n-i} \right)X_i \right)\\
  &= p(1-p) \left( \left( \frac{1}{n} - \lambda^{n-1} \right)^2 + \sum_{i=2}^n \left( \frac{1}{n} - (1-\lambda) \lambda^{n-i} \right)^2\right)
\end{align*}
\end{proof}

Please note that $\text{Var}(\hat{p}^{\lambda}_n - \bar{p}_n)$ can be computed recursively such that all variances up to $n$ can be computed in $O(n)$ time. The actual recursions are not shown, but are straightforward to compute from \eqref{eq:5}. Another appealing property is that the variance does not depend on the stream of observations and can be computed before the data streaming starts. This lays the foundations for building very efficient algorithms.

Theorem \ref{the:1} stated the expectation and variance of the distribution of $\hat{p}^{\lambda}_n - \bar{p}_n$. Next we investigate other properties the distribution. From the proof we saw that $\hat{p}^{\lambda}_n - \bar{p}_n$ could be written as follows
\begin{align*}
  \hat{p}^{\lambda}_n - \bar{p}_n = \left( \frac{1}{n} - \lambda^{n-1} \right)X_1 + \sum_{i=2}^n \left( \frac{1}{n} - (1-\lambda) \lambda^{n-i} \right)X_i
\end{align*}
which is a weighted sum of the independent Bernoulli variables. If the sum satisfies the Lindeberg criterion (and thus the Lyapunov criterion), the sum will, according to the central limit Theorem, converge to a normal distribution \cite{karr12}. Unfortunately the sum does not satisfy this criterion (proofs omitted). A second option is to study the distribution of $\hat{p}^{\lambda}_n - \bar{p}_n$ by stochastic simulation. We perform the following experiment. We generated $n=50$ independent outcomes from the Bernoulli distribution and computed $\hat{p}^{\lambda}_{50} - \bar{p}_{50}$ using $\lambda = 0.95$. Further we repeated this procedure $N=10\,000$ times. The upper left panel in Figure \ref{fig:1} shows the histogram of these values when $p = 0.1$.
\begin{figure}
  \centering
  \includegraphics[width = 11cm]{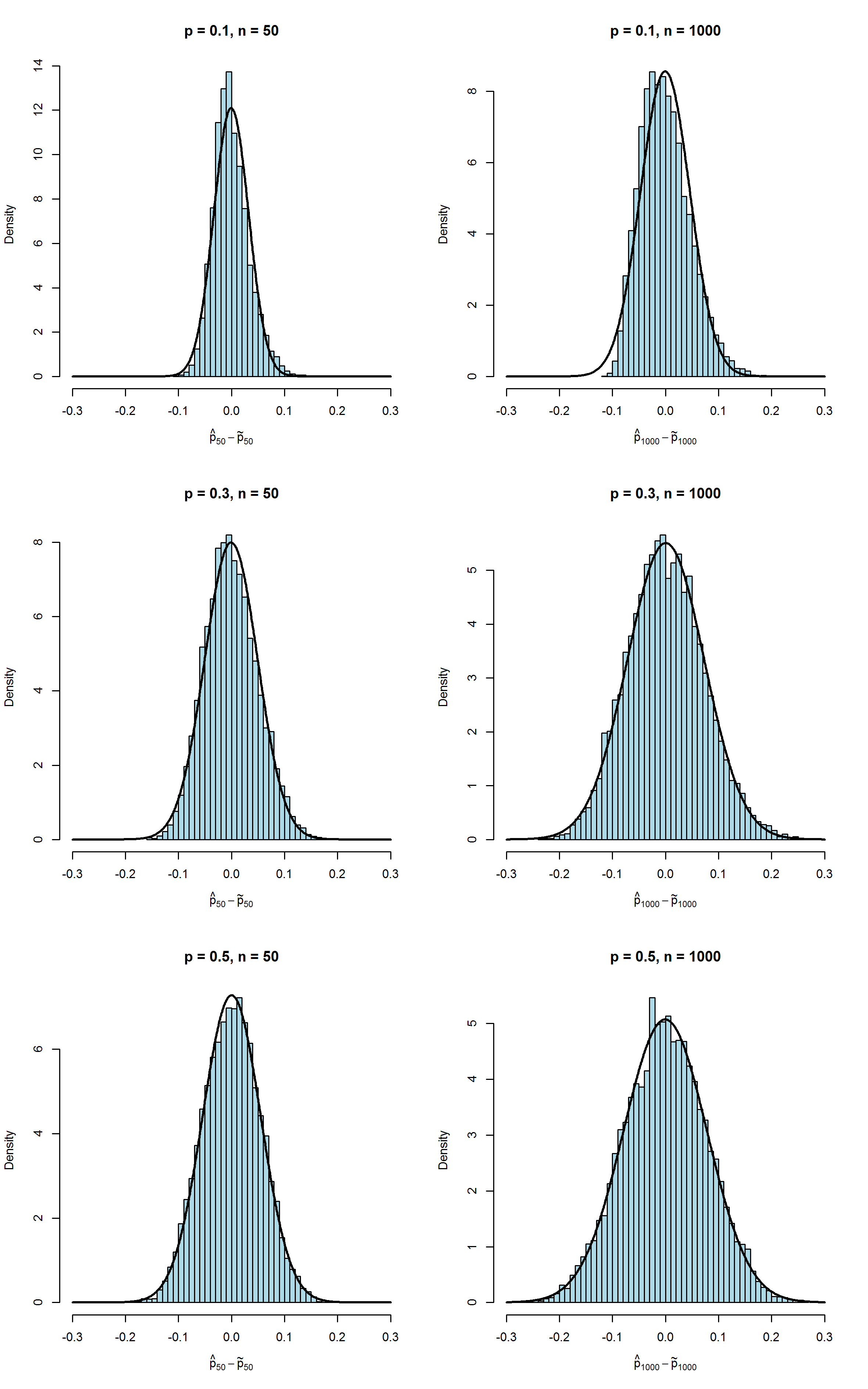}
  \caption{The distribution of $\hat{p}^{\lambda}_{n} - \bar{p}_{n}$ and the normal distribution with the same expectation and variance (black curve).}
  \label{fig:1}
\end{figure}
The black curve is the normal distribution with expectation and variance as given by Theorem \ref{the:1}. The upper right panel shows the same, but with $n = 1000$. The second and the third row shows the same but with $p = 0.3$ and $p = 0.5$. Overall we see that the distribution of $\hat{p}^{\lambda}_{n} - \bar{p}_{n}$ is almost identical to a normal distribution. We only observe that when $p$ is small (or high) the distribution is a little asymmetric compared to a normal distribution. Based on these observations it is a reliable choice to build a test assuming that $\hat{p}^{\lambda}_{n} - \bar{p}_{n}$ is normally distributed. We then get the following test.
\begin{theorem}
\label{the:2}
 Let $X_1, X_2, X_3, \ldots$ represent a stream of independent and identically distributed Bernoulli stochastic variables with parameter $p$. Further let $z_{\alpha}$ denote the $\alpha$ quantile of the standard normal distribution. Define the hypotheses
 \begin{itemize}
 \item[$H_0$:] The underlying $p$ has not changed value
 \item[$H_1$:] The underlying $p$ has changed value
 \end{itemize}
Suppose that we decide to reject $H_0$ if
\begin{align}
  \label{eq:6}
  \frac{|\hat{p}^{\lambda}_{n} - \bar{p}_{n}|}{\sqrt{\text{Var}(\hat{p}^{\lambda}_{n} - \bar{p}_{n})}} > z_{\alpha/2}
\end{align}
Then the probability of rejecting $H_0$ if $H_0$ is true is approximately $\alpha$ and the rejection rule \eqref{eq:6} controls the type I error.
\end{theorem}
\begin{proof}
Let $N(\mu, \sigma)$ denote a normal distribution with expectation $\mu$ and standard deviation $\sigma$. From the discussion above and Figure \ref{fig:1} we know that
\begin{align*}
  \frac{\hat{p}^{\lambda}_{n} - \bar{p}_{n}}{\sqrt{\text{Var}(\hat{p}^{\lambda}_{n} - \bar{p}_{n})}} \approx N(0,1)
\end{align*}
which means that
\begin{align*}
  P(\text{Type I  error}) &= P(\text{reject }H_0 \, | \, H_0\text{ true}) =\\ & =P\left(  \dfrac{|\hat{p}^{\lambda}_{n} - \bar{p}_{n}|}{\sqrt{\text{Var}(\hat{p}^{\lambda}_{n} - \bar{p}_{n})}} > z_{\alpha/2} \right) \approx \alpha
\end{align*}
\end{proof}
From Theorem \ref{the:1} we see that $\text{Var}(\hat{p}^{\lambda}_{n} - \bar{p}_{n})$ depends on $p$ which of course is unknown. To perform the test above, a natural choice is to substitute $p$ with the sample mean $\bar{p}_{n}$ since this is our best estimate of $p$ under the hypothesis that $p$ is constant.

The basic idea of our method is to estimate $p$ using the sample mean, but perform occasional jumps if the test in Theorem \ref{the:2} brings evidence that $p$ has switched value. In the Section we discuss different alternatives to perform the jumps.

\subsection{Performing a jump}
\label{sec:jump}

Let $\tilde{p}_n$ denote the estimate using the sample mean with jumps method after the arrival of $X_n$. Further let $\tilde{\lambda}_n$ denote the value used for $\lambda_n$ in the recursions in \eqref{eq:3} to compute $\tilde{p}_n$. Assume now that the test in Theorem \ref{the:2} brings evidence that $p$ has switched value which means that the current estimate $\tilde{p}_n$ is not reliable (since it is based on the sample mean). Therefore we need to adjust the estimate $\tilde{p}_n$ (jump). Two options seem natural to perform the jump.
\begin{itemize}
\item Forget the whole estimation history and set $\tilde{p}_{n} = X_n$
\item Assume that the current estimate based on the weak estimator with constant $\lambda$, $\hat{p}_n^{\lambda}$, is reliable since it adjusts fast and set $\tilde{p}_{n} = \hat{p}_n^{\lambda}$.
\end{itemize}
A third option could be to set $\tilde{p}_{n}$ equal to some weighting between these two alternatives.

To continue the update of the estimator $\tilde{p}_n$ after the jump, we also need to decide a new value for $\tilde{\lambda}_n$. There are at least two natural alternatives
\begin{itemize}
\item Recall that by setting $\lambda_n = (n-1)/n$ in \eqref{eq:3}, we get the sample mean. If we decide to follow the first option above and set $\tilde{p}_{n} = X_n$, $\tilde{p}_{n}$ is just the sample mean of one observation which means that it is natural to set  $\tilde{\lambda}_n = (1-1)/1 = 0$.
\item If we decide to follow the second option above and set $\tilde{p}_{n} = \hat{p}_n^{\lambda}$, it seems natural to do the next update of $\tilde{p}_n$ similar to the update $\hat{p}_n^{\lambda}$, which means to relate $\tilde{\lambda}_n$ to $\lambda$. Since we will continue to update $\tilde{p}_n$ according to the sample mean, we must relate such a choice to the number of terms in a sample mean. We do this as follows. Define $\tilde{n}$ as the solution of the equation
  \begin{align*}
    \frac{\tilde{n}-1}{\tilde{n}} = \lambda
  \end{align*}
Solving with respect to $\tilde{n}$ and rounding off to the nearest integer we get
\begin{align*}
  \tilde{n} = \left[ \frac{1}{1-\lambda} \right]
\end{align*}
where $[a]$ denotes the value of a $a$ rounded of to the nearest integer. The interpretation of $\tilde{n}$ is the number of terms in a sample mean in which an update of the estimate will be similar to the weak estimator $\hat{p}_n^{\lambda}$.
\end{itemize}
Note that the choice of $\tilde{\lambda}_n$ in the first alternative above is equivalent to setting $\tilde{n} = 1$. It may be that when the test in Theorem \ref{the:2} detects a change in $p$, the value of $\hat{p}_n^{\lambda}$ has not converged completely around the new value of $p$. Therefore a value of $\tilde{n}$ somewhere between 1 and $[1/(1-\lambda)]$ may be an even better alternative. By relating the variance of $\tilde{p}_n$ to a sample mean with $\tilde{n}$ terms, the variance $\text{Var}(\hat{p}^{\lambda}_{n} - \tilde{p}_{n})$, which we need in the test in Theorem \ref{the:2}, can be computed recursively. In addition, all the variances can recursively be computed in advance before the data stream starts.

Before the algorithm can be run, we also need to decide a value for $\alpha$ in the test proposed in Theorem \ref{the:2}. When we run the test, the probability of wrongly detecting a change in $p$, is approximately $\alpha$. In practice we may run the test many times, for example every tenth iteration. If we run the test many times, the chance of wrongly detecting a change in $p$ in some of these tests naturally will be larger then $\alpha$. This refers to the multiple testing problem in the statistical literature, see e.g. \cite{benjamini95}. A simple and much used approach is the Bonferroni correction where a significance level of $\alpha/M$ is used instead of $\alpha$, where $M$ is the number of tests. There are two challenges with applying this approach (and other standard corrections). First, we do not know the number of tests we need to run. Second, the Bonferroni correction assumes that all the tests are independent. In our case this is far from true, since two subsequent tests are based on almost the same data stream (only a few extra observation have been added since the last test) and the outcomes are highly correlated. Using the Bonferroni correction will result in a too low significance level, and the tests may never detect that $p$ has changed. In practice, setting $\alpha$ to about $10^{-3}$ overall performs well and is, as expected, somewhere between standard significance levels (0.05) and Bonferroni corrected levels.

The algorithm using the second option above is shown in Algorithm \ref{alg:1}.
\begin{algorithm}
\caption{The sample mean with jumps algorithm.}\label{alg:1}
\INPUT\\
$X_1, X_2, X_3, \ldots$ //Stream of Bernoulli variables\\
$\lambda$\\
$\alpha$\\
$D$ //How often to perform the test in Theorem \ref{the:2}\\
$N$ //Max number of iterations \\
\METHOD
\begin{algorithmic}[1]
\STATE $\tilde{n} \gets 0$
\STATE $\hat{p}^{\lambda}_1 \gets X_1$
\STATE $\tilde{p}_1 \gets X_1$
  \FOR {$n \in 1,2,\ldots,N$}
    \STATE $\hat{p}_n^{\lambda} \gets \lambda \hat{p}^{\lambda}_{n-1} + (1 - \lambda)X_n$
    \STATE $\tilde{n} \gets \tilde{n}+1$
    \STATE $\tilde{p}_n \gets \frac{\tilde{n}-1}{\tilde{n}} \tilde{p}_{n-1} + \frac{1}{\tilde{n}} X_n$
    \IF {$n \bmod D == 0$}
      \IF {$\frac{|\hat{p}^{\lambda}_{n} - \tilde{p}_{n}|}{\sqrt{\text{Var}(\hat{p}^{\lambda}_{n} - \tilde{p}_{n})}} > z_{\alpha/2}$}
        \STATE $\tilde{p}_n \gets \hat{p}_n^{\lambda}$
        \STATE $\tilde{n} \gets [1/(1-\lambda)]$
      \ENDIF
    \ENDIF
  \ENDFOR
\end{algorithmic}
\end{algorithm}

\section{Extension to the multinomial case}
\label{sec:multinom}

We now show how the jump algorithm above can be extended to the multinomial case. As described above, a Bernoulli variable takes the values 0 or 1 with probabilities $1-p$ and $p$, respectively. For the multinomial case this is extended such that $X$ takes one of the values $\{1,2,\ldots,r\}$ with probabilities $\{p_1, p_2, \ldots, p_r\}$, such that $\sum_{i=1}^r p_i = 1$.
For ease of presentation below, define a stochastic vector $Y$ which is a map from $X$ as follows
\begin{align}
  \label{eq:12}
  Y = [\mathbb{I}(X = 1), \mathbb{I}(X = 2), \ldots, \mathbb{I}(X = r)]
\end{align}
where $\mathbb{I}(A)$ denote the indicator function returning one if $A$ is true and zero if $A$ is false. We see that $Y$ is a vector with value one in position $X$ and zero in all the other positions.

Let $Y_1, Y_2, Y_3, \ldots$ denote a stream of independent stochastic variables identical to $Y$. We now want to maintain running estimates of the probabilities $\{p_1, p_2, \ldots, p_r\}$. The SLWE in \eqref{eq:3} can easily be extend to the multinomial case as follows
\begin{align}
  \label{eq:7}
  [\hat{p}_{n,1}, \ldots, \hat{p}_{n,r}] = \lambda_n[\hat{p}_{n-1,1}, \ldots, \hat{p}_{n-1,r}] + (1 - \lambda_n) Y_n
\end{align}
where $\hat{p}_{i,n}$ denote the estimate of $p_i$ after receiving the variable $Y_n$ from the data stream.

Now let $[\hat{p}^{\lambda}_{n,1}, \ldots, \hat{p}^{\lambda}_{n,r}]$ denote estimates based on \eqref{eq:7} using a constant value of $\lambda$ and let $[\overline{p}_{n,1}, \ldots, \overline{p}_{n,r}]$ denote the sample mean, i.e. using $\lambda_n = (n-1)/n$. Following the same argumentation as in Section \ref{sec:shiftenv} we have the following
\begin{align}
  \label{eq:10}
  \frac{\hat{p}^{\lambda}_{n,i} - \bar{p}_{n,i}}{\sqrt{\text{Var}(\hat{p}^{\lambda}_{n,i} - \bar{p}_{n,i})}} \approx N(0,1), \,\,\, i = 1,2,\ldots,r
\end{align}
As an extension to Section \ref{sec:shiftenv}, we now want to construct a statistical test to check wether the unknown probability vector $[p_1, p_2, \ldots, p_r]$ has changed value. A common statistical test on the probability vector of the multinomial distribution is the Pearson's $\chi^2$ test \cite{agresti2011categorical}. Adapting the $\chi^2$ test to the application in this paper, we get the following theorem.
\begin{theorem}
\label{the:3}
 Let $Y_1, Y_2, Y_3, \ldots$ represent a stream of independent and identically distributed multinomial stochastic vectors with probability vector $[p_1, p_2, \ldots, p_r]$. Further let $\chi^2_{n,\alpha}$ denote the $\alpha$ quantile of the $\chi^2$ distribution with $n$ degrees of freedom. Define the hypotheses
 \begin{itemize}
 \item[$H_0$:] The underlying probability vector $[p_1, p_2, \ldots, p_r]$ has not changed value
 \item[$H_1$:] The underlying probability vector $[p_1, p_2, \ldots, p_r]$ has changed value
 \end{itemize}
Suppose that we decide to reject $H_0$ if
\begin{align}
  \label{eq:9}
  \sum_{i=1}^r  \frac{\left(\hat{p}^{\lambda}_{n,i} - \bar{p}_{n,i}\right)^2}{\text{Var}(\hat{p}^{\lambda}_{n,i} - \bar{p}_{n,i})} > \chi^2_{r-1,\alpha}
\end{align}
Then the probability of rejecting $H_0$ if $H_0$ is true is approximately $\alpha$ and the rejection rule \eqref{eq:9} controls the type I error.
\end{theorem}
\begin{proof}
It is well known that the sum of $n$ independent squared standard normally distributed stochastic variables is $\chi^2_n$ distributed, denoting a $\chi^2$ distribution with $n$ degrees of freedom. From \eqref{eq:10} we see that the sum in \eqref{eq:9} is a sum of approximately squared standard normally distributed stochastic variables and therefore is approximately $\chi^2$ distributed. Knowing $r-1$ terms in the sum, the last term can be computed since the probability estimates sum to one. The sum in \eqref{eq:9} thus is approximately $\chi^2_{r-1}$ distributed
\begin{align}
  \label{eq:11}
  \sum_{i=1}^r  \frac{\left(\hat{p}^{\lambda}_{n,i} - \bar{p}_{n,i}\right)^2}{\text{Var}(\hat{p}^{\lambda}_{n,i} - \bar{p}_{n,i})} \approx \chi^2_{r-1}
\end{align}
Theorem \ref{the:3} follows directly from \eqref{eq:11}.
\end{proof}
Algorithm \ref{alg:2} shows the resulting jump algorithm for the multinomial case.
\begin{algorithm}
\caption{The sample mean with jumps algorithm for the multinomial case.}\label{alg:2}
\INPUT\\
$Y_1, Y_2, Y_3, \ldots$ //Stream of multinomial variables on vector form (recall Eq. \eqref{eq:12})\\
$\lambda$\\
$\alpha$\\
$D$ //How often to perform the test in Theorem \ref{the:2}\\
$N$ //Max number of iterations \\
\METHOD
\begin{algorithmic}[1]
\STATE $\tilde{n} \gets 0$
\STATE $[\hat{p}^{\lambda}_{1,1}, \ldots, \hat{p}^{\lambda}_{1,r}] \gets Y_1$
\STATE $[\tilde{p}_{1,1}, \ldots, \tilde{p}_{1,r}] \gets Y_1$
  \FOR {$n \in 1,2,\ldots,N$}
    \STATE $[\hat{p}^{\lambda}_{n,1}, \ldots, \hat{p}^{\lambda}_{n,r}] \gets \lambda [\hat{p}^{\lambda}_{n,1}, \ldots, \hat{p}^{\lambda}_{n,r}] + (1 - \lambda)Y_n$
    \STATE $\tilde{n} \gets \tilde{n}+1$
    \STATE $[\tilde{p}_{n,1}, \ldots, \tilde{p}_{n,r}] \gets \frac{\tilde{n}-1}{\tilde{n}} [\tilde{p}_{n-1,1}, \ldots, \tilde{p}_{n-1,r}] + \frac{1}{\tilde{n}} Y_n$
    \IF {$n \bmod D == 0$}
      \IF {$\sum_{i=1}^r  \frac{\left(\hat{p}^{\lambda}_{n,i} - \tilde{p}_{n,i}\right)^2}{\text{Var}(\hat{p}^{\lambda}_{n,i} - \tilde{p}_{n,i})} > \chi^2_{r-1,\alpha}$}
        \STATE $[\tilde{p}_{n,1}, \ldots, \tilde{p}_{n,r}] \gets [\hat{p}^{\lambda}_{n,1}, \ldots, \hat{p}^{\lambda}_{n,r}]$
        \STATE $\tilde{n} \gets [1/(1-\lambda)]$
      \ENDIF
    \ENDIF
  \ENDFOR
\end{algorithmic}
\end{algorithm}

\section{Experiments}
\label{sec:exp}

In this Section we evaluate the methodology above for both synthetic and real-life data. In all the experiments reported below, we set $\tilde{p}_n = \hat{p}_n^{\lambda}$ after a jump, i.e. the second alternative discussed in Section \ref{sec:jump} and as given in Algorithms \ref{alg:1} and \ref{alg:2}. We have not found any paper in the literature dealing with tracking the probabilities in binomial and multinomial distributions in abruptly changing environments. The most related papers, in our opinion, are \cite{ross2012} and \cite{gama2004learning}, but they consider a slightly different problem, namely the problem of concept drift. Those methods may be modified to accommodate online estimation as for the devised algorithm in this paper, but we have not looked into that. Therefore, in our experiments we compare the performance of the suggested algorithms in this paper with the SLWE estimator from \cite{OommenRueda06}, i.e. \eqref{eq:2} with constant lambda $\lambda_n = \lambda$ and denote it $\hat{p}_n^{\lambda}$.

\subsection{Synthetic data example}

We will evaluate the binomial case (Algorithm \ref{alg:1}) and the multinomial case for $r = 4$ classes (Algorithm \ref{alg:2}). Figure \ref{fig:2} shows a comparison of the different estimators for the binomial case when the changes in $p$ are large.
\begin{figure}
  \centering
  \includegraphics[width = 0.5\textwidth]{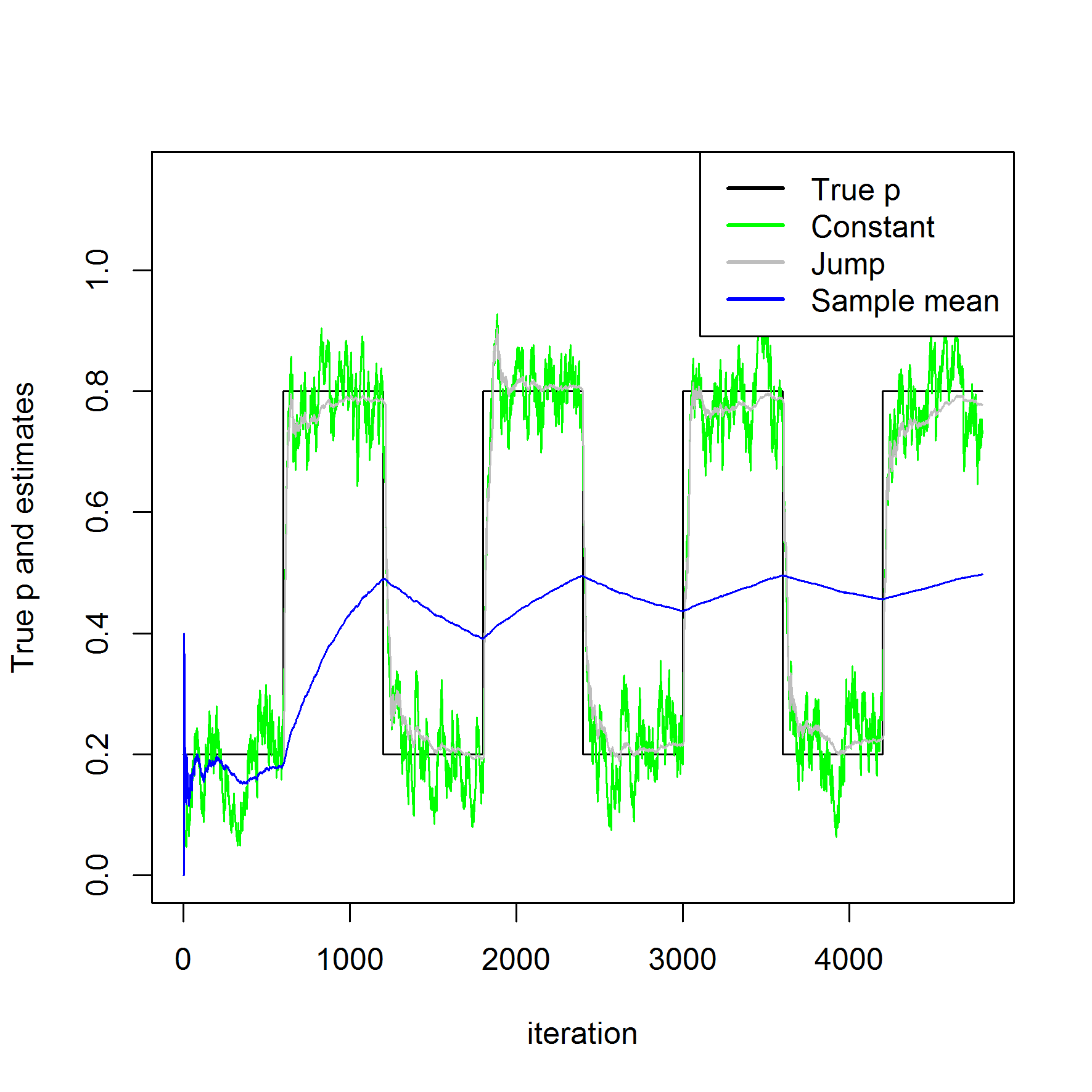}
  \caption{Evaluation of the estimators in an environment with large jumps in $p$. The black curve shows the true $p$ in each iteration. The gray, green and blue curves show the estimators $\tilde{p}_n$, $\hat{p}_n^{\lambda}$ and the sample mean, respectively.}
  \label{fig:2}
\end{figure}
The gray, green and blue curves show the jump estimator ($\tilde{p}_n$), the SLWE with constant $\lambda$ ($\hat{p}_n^{\lambda}$) and the sample mean, respectively. The black curve shows the true value of $p$ in each iteration. We see that the test in Theorem \ref{the:2} detects the changes in $p$ very efficiently such that on average will $\tilde{p}_n$ (gray) perform better then $\hat{p}_n^{\lambda}$. We also see, as expected, that the sample mean is not very useful in a dynamic environment.

Figure \ref{fig:3} shows results from a similar experiment, but where the changes in $p$ are smaller.
\begin{figure}
  \centering
  \includegraphics[width = 0.5\textwidth]{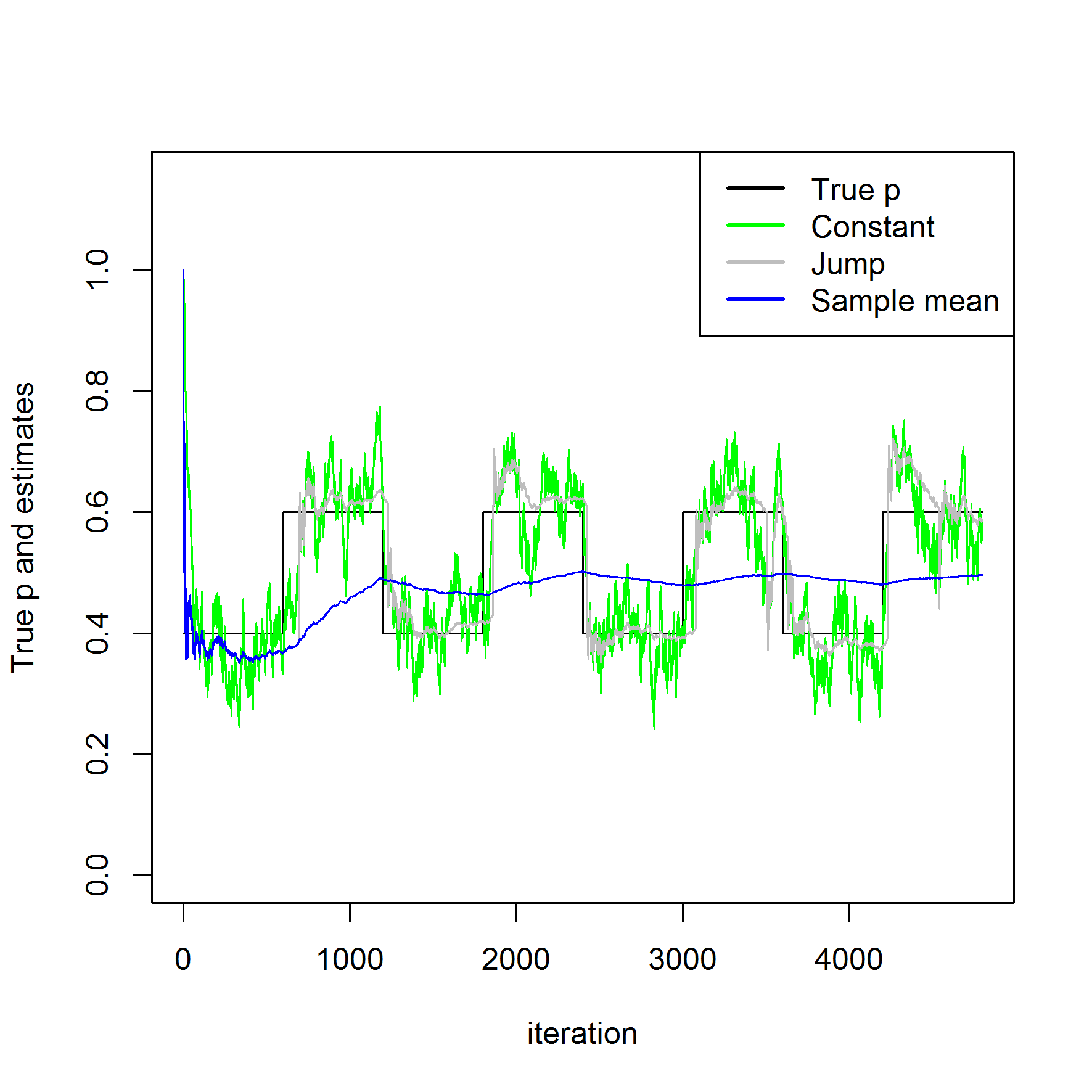}
  \caption{Evaluation of the estimators in an environment with small jumps in $p$. The black curve shows the true $p$ in each iteration. The gray, green and blue curves show the estimators $\tilde{p}_n$, $\hat{p}_n^{\lambda}$ and the sample mean, respectively.}
  \label{fig:3}
\end{figure}
We see that the changes in $p$ also here will be efficiently detected and that $\tilde{p}_n$ performs better than $\hat{p}_n^{\lambda}$. In both experiments above we chose $\alpha = 10^{-3}$, $\lambda = 0.96$ and $\tilde{n} = [1/(1-\lambda)] = 25$.

In Figure \ref{fig:4} we also evaluate the jump estimator for an environment where $p$ is changing smoothly. More specifically the true $p$ changes following a cosine function. For such an environment, it seems like $\tilde{p}_n$ and $\hat{p}_n^{\lambda}$ perform almost equally well. Note that even though the jump estimator is not constructed for such environments, but we see that it still performs well. In this experiment we chose $\alpha = 10^{-2}$, $\tilde{n} = 1$ and still $\lambda = 0.96$
\begin{figure}
  \centering
  \includegraphics[width = 0.5\textwidth]{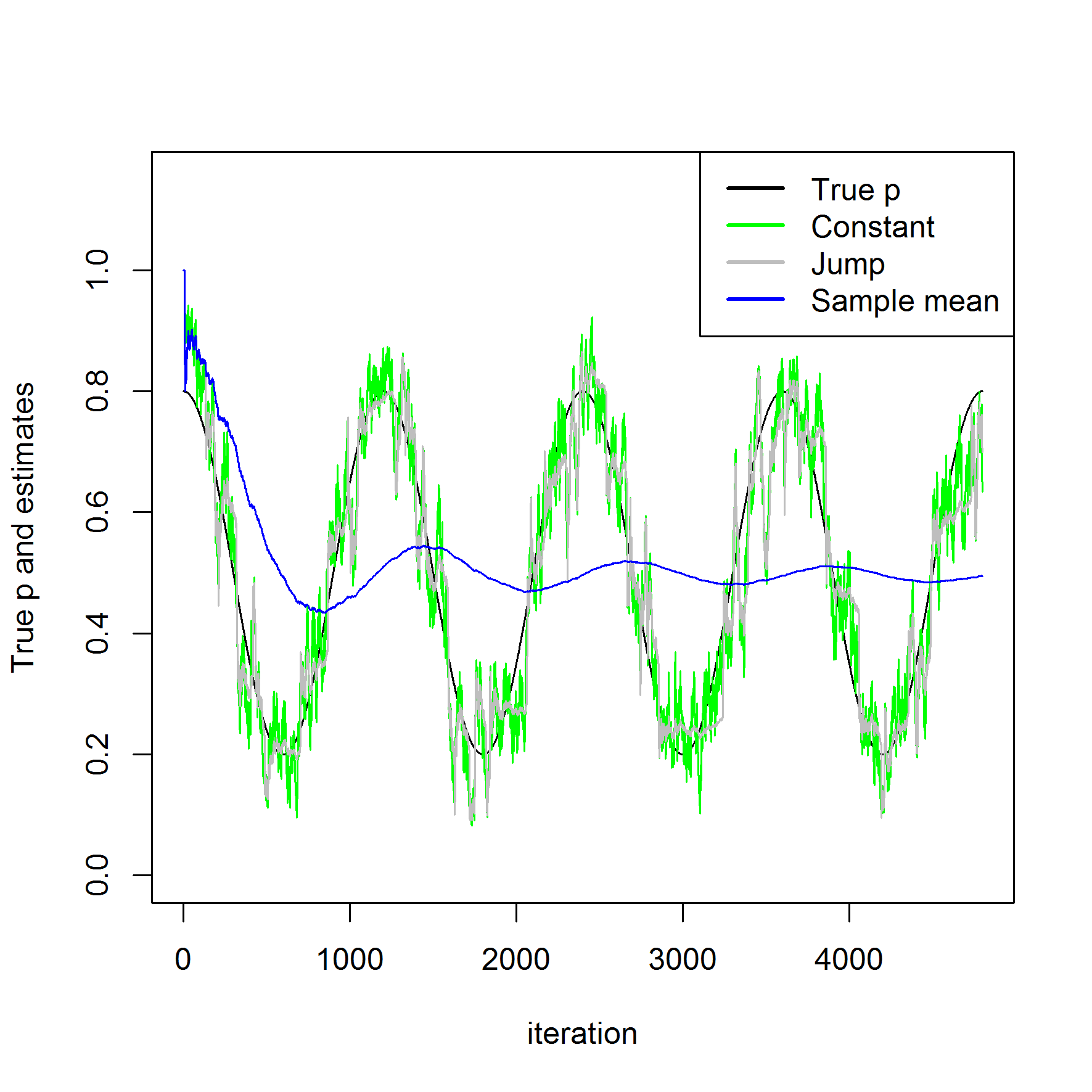}
  \caption{Evaluation of the estimators in an environment where $p$ dynamically is changing. The black curve shows the true $p$ in each iteration. The gray, green and blue curves show the estimators $\tilde{p}_n$, $\hat{p}_n^{\lambda}$ and the sample mean, respectively.}
  \label{fig:4}
\end{figure}

For the jump algorithms described by Algorithms \ref{alg:1} and \ref{alg:2} there are three tuning parameters, namely $\alpha$, $\tilde{n}$ and $\lambda$. We now want to chose values for these parameters such that the estimation error will be as small as possible. We measure the estimation error as the difference in absolute value between the true $p$ and the estimator averaged over all the iterations.

We start by investigating reasonable values for $\alpha$. For the binomial and multinomial algorithm we computed the estimation error for different choices of $z_{\alpha}$ and $\chi^2_{r-1,\alpha}$, respectively. To reduce the Monte Carlo error we ran the Bernoulli and multinomial data streams for  $5\cdot 10^6$ iterations. In the experiments we used $\lambda = 0.95$ and $\tilde{n} = [1/(1-\lambda)] = 20$. We assumed that the system shifted state every $D = 600$ iteration similar to the examples in Figures \ref{fig:2} and \ref{fig:3}.

For the binomial case we considered three different cases.
\begin{itemize}
\item Large changes: $p$ changed with time as shown in Figures \ref{fig:2}.
\item Small changes: $p$ changed with time as shown in Figures \ref{fig:3}.
\item Dynamic: $p$ changed with time as shown in Figures \ref{fig:4}.
\end{itemize}
For the multinomial case we considered two different cases.
  \begin{itemize}
  \item[1.] Every $D = 600$ iterations, we changed the probability vector as follows
    \begin{itemize}
    \item Draw a random number $\rho$ uniformly from $1,2,\ldots,r$
    \item Set $p_{\rho} = 0.8$
    \item Set $p_i = 0.2/(r-1)$ for $i \neq \rho$
    \end{itemize}
    Below we refer to this alternative as 'spike probability'.
  \item[2] Every $D = 600$ iterations, we updated the probability vector as an outcome from the Dirichlet distribution with parameter values $\alpha_1 = 1, \ldots, \alpha_r = 1$. This is referred to as the flat Dirichlet distribution and the probability distribution is uniformly distributed over the simplex of possible probability vectors, i.e. the vectors satisfying $p_1, p_2, \ldots, p_r > 0$ and $\sum_{i=1}^r p_i = 1$. Below we refer to this alternative as 'flat Dirichlet'.
  \end{itemize}
The results are shown in Figure \ref{fig:5}.
\begin{figure}
  \begin{tabular}{cc}
  \includegraphics[width = 0.5\textwidth]{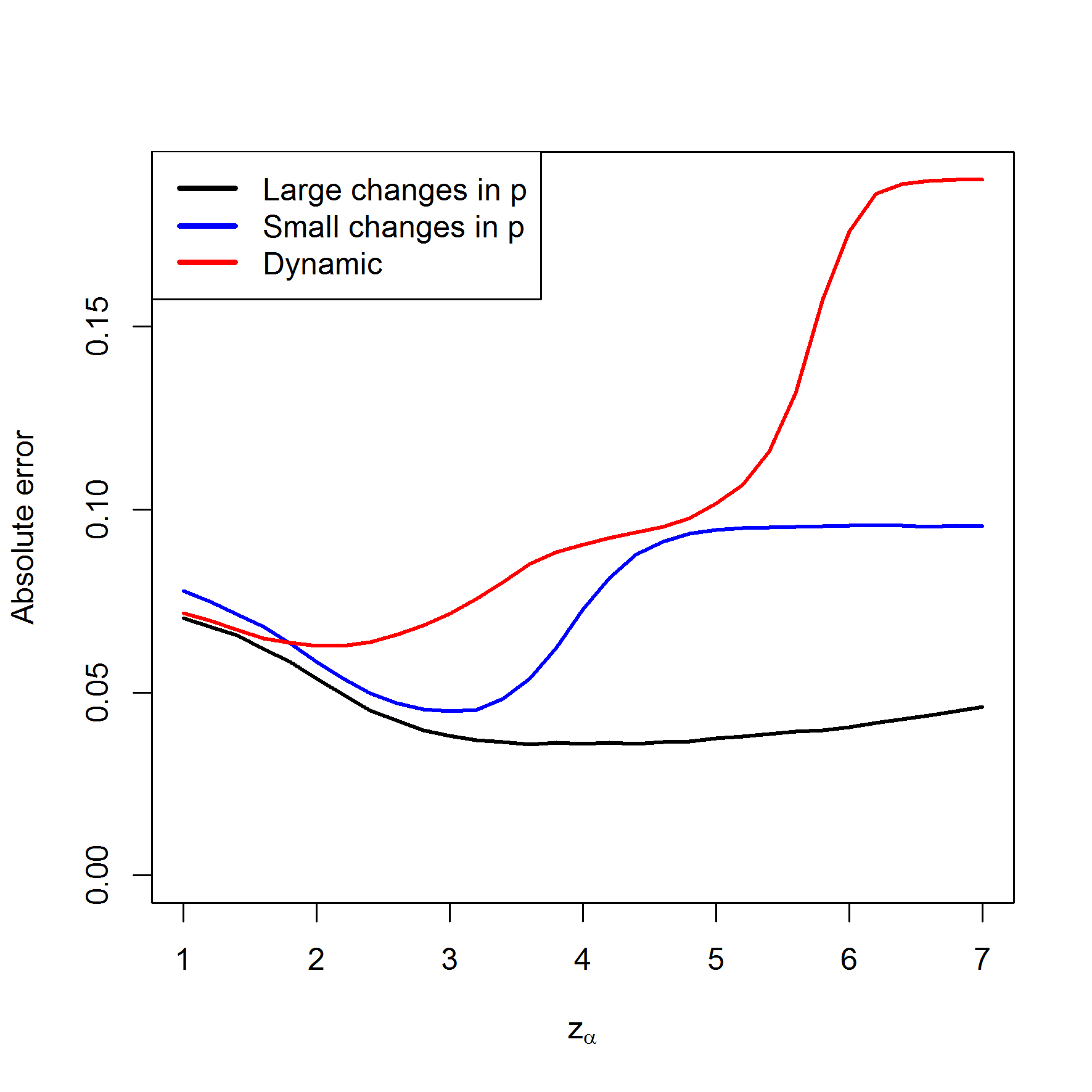} & \includegraphics[width = 0.5\textwidth]{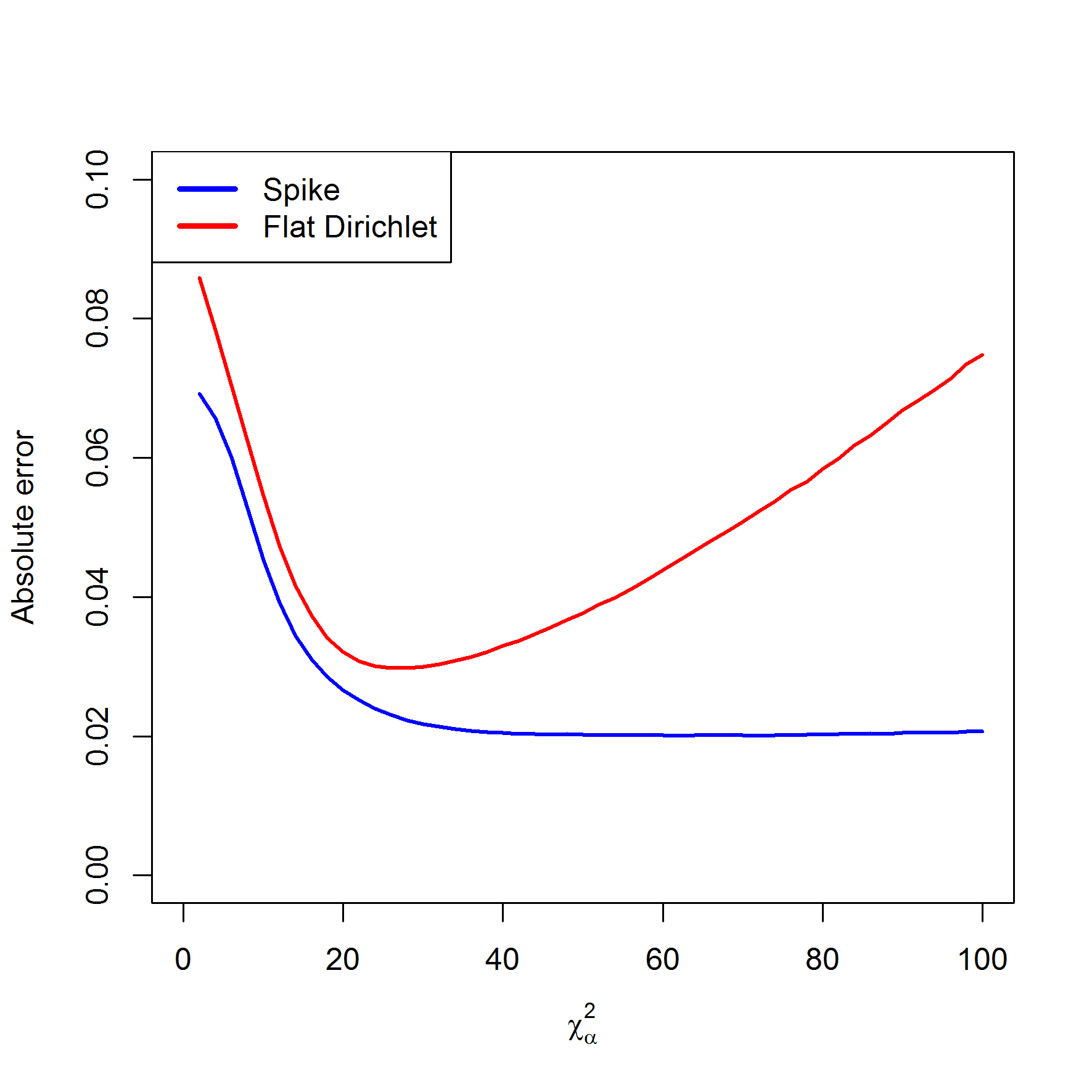}
  \end{tabular}
  \caption{The left and and right panels show estimation error as a function of $z_{\alpha}$ (binomal case) and $\chi^2_{r-1.\alpha}$ (multinomial case), respectively. For the left panel, the black, blue and red curves refer to experiments where the changes in $p$ are large, small and dynamic. For the right panel, the blue and red curves refer to experiments where the changes in $p$ are based on spike probability and flat Dirichlet alternatives, respectively.}
  \label{fig:5}
\end{figure}
We start discussing the binomial case. For the blue curve in the left panel of Figure \ref{fig:5} we see that an optimal value of $z_{\alpha}$ is about 3 which is equivalent to $\alpha \approx 10^{-3}$. By choosing smaller values of $z_{\alpha}$, the test will too often wrongly detect changes. Choosing a too high value of $z_{\alpha}$, the test will detect changes in $p$ too late or never. With $z_{\alpha}$ above 5, the test will never detect the changes in $p$ and we reach a limit in the estimation error which is equal to the estimation error using the sample mean. When the changes in $p$ are large (black curve), we can allow using higher values of $z_{\alpha}$ since we still are able to detect the large changes in $p$. An optimal value for $z_{\alpha}$ is around 4. Choosing an even higher value of $z_{\alpha}$ slightly reduces the performance because the method uses a few iterations more before detecting that $p$ has changed value. For the dynamic system, we see that the estimation error is higher and is as expected since the method in this paper is not directly constructed for such environment. We see that the optimal value for $z_{\alpha}$ is around 2.2 which seems reasonable. Since $p$ continuously is changing value, $z_{\alpha}$ should not be chosen too high to be able to keep track of these changes.

For the multinomial case (right panel), we see that for the flat Dirichlet alternative, an optimal value is $\chi^2_{r-1,\alpha} \approx 25$ which is equivalent to $\alpha \approx 10^{-5}$. For the spike probability alternative any value of $\chi^2_{r-1,\alpha}$ between 25 and 100 perform well. We see that a lower value of $\alpha$ performs the best in the multinomial cases compared to the binomial cases. The reason is that it is easier to detect a change in the probability vector in the multinomial case compared to  the binomial case.

We turn our attention now to evaluating the optimal values for $\tilde{n}$. The results are shown in Figure \ref{fig:6}. Also in this experiment we set $\lambda = 0.95$. Further we sat $z_{\alpha} = 3$ and $\chi^2_{r-1,\alpha} \approx 25$.
\begin{figure}
  \centering
  \begin{tabular}{cc}
  \includegraphics[width = 0.5\textwidth]{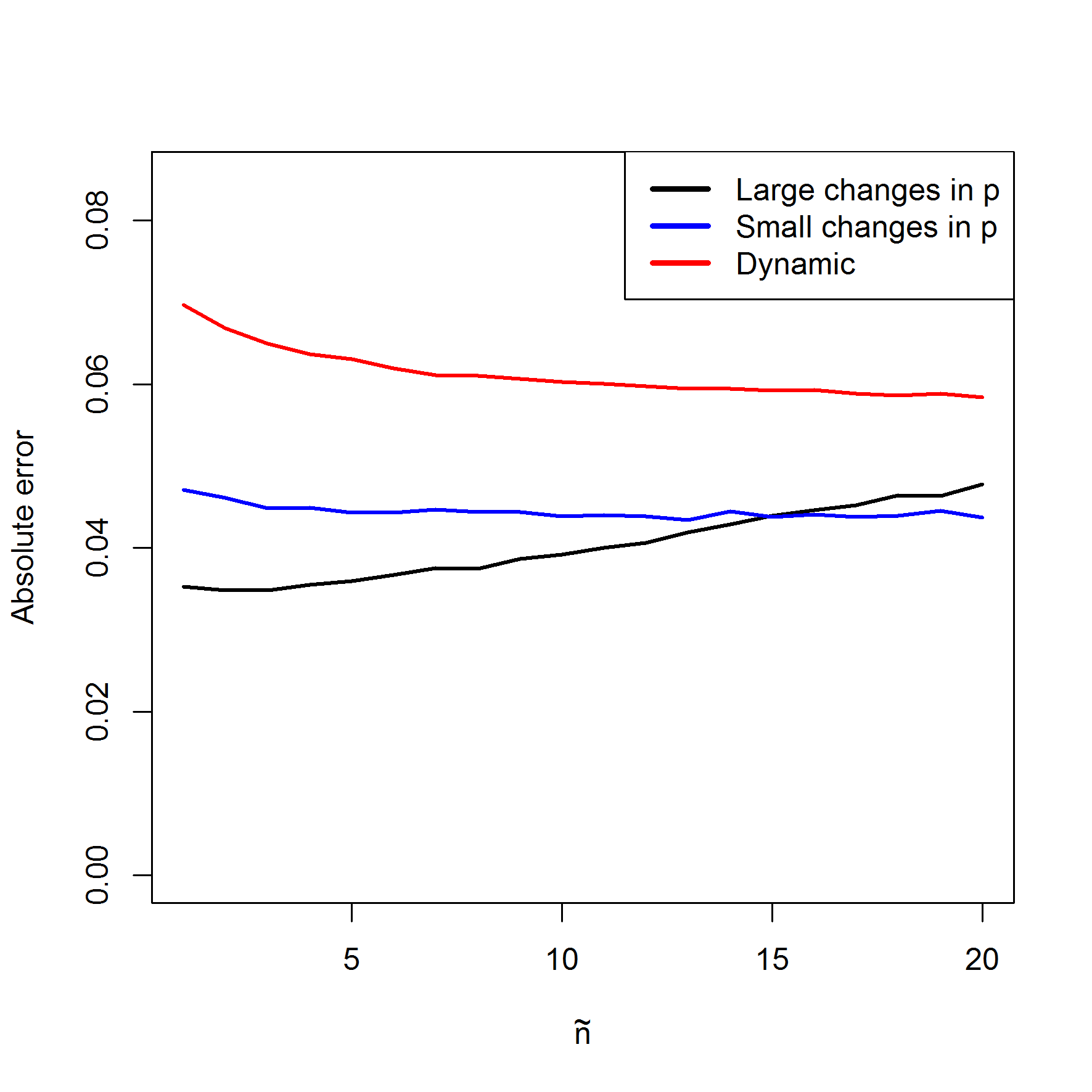} & \includegraphics[width = 0.5\textwidth]{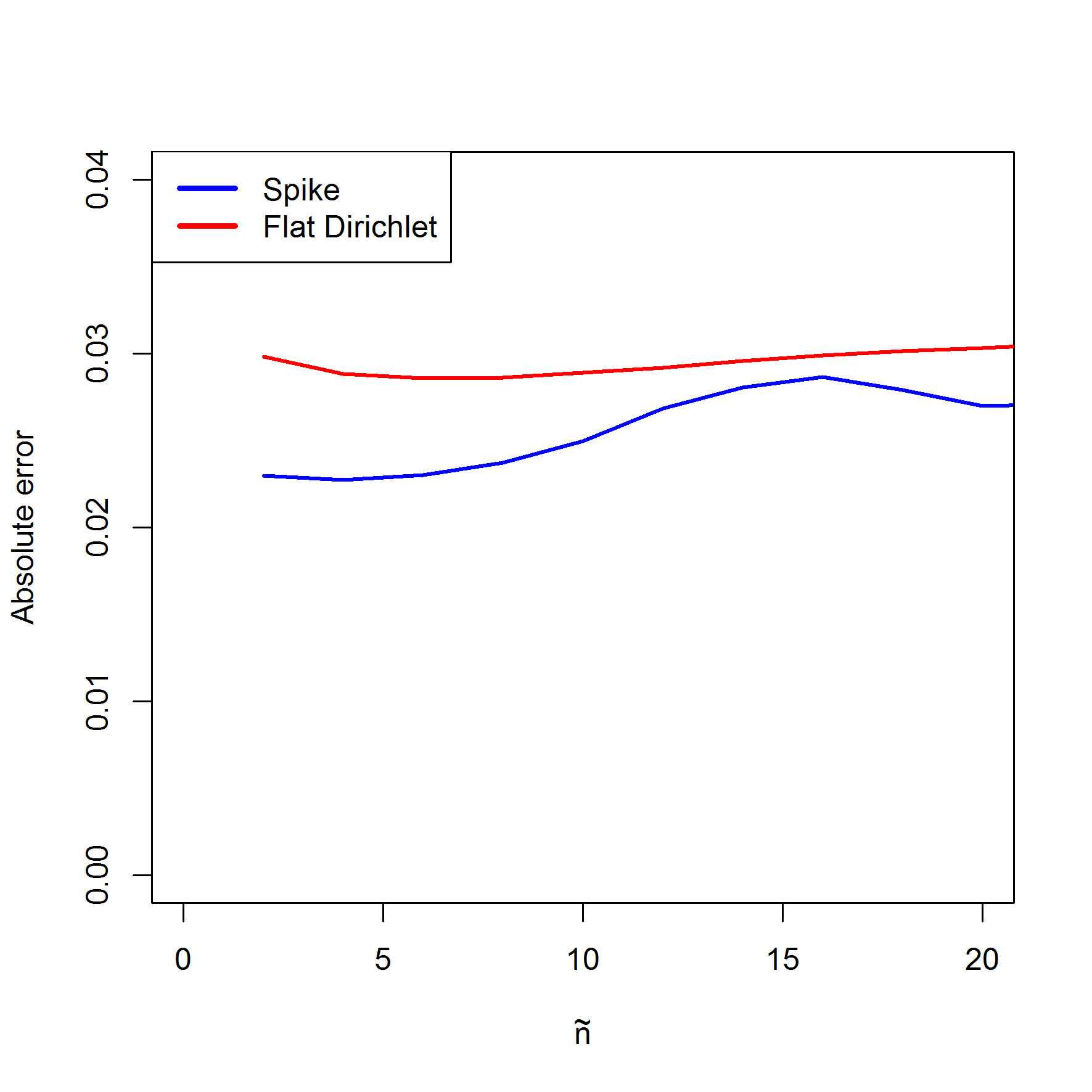}
  \end{tabular}
  \caption{Estimation error as a function of $\tilde{n}$. For the left panel, the black, blue and red curves refer to experiments where the changes in $p$ are large, small and dynamic. For the right panel, the blue and red curves refer to experiments where the changes in $p$ are based on spike probability and flat Dirichlet alternatives, respectively.}
  \label{fig:6}
\end{figure}
Overall we see that the estimation error does not depend strongly on the choice of $\tilde{n}$. Please note that the increase in estimation error for $\tilde{n} \approx 15$ for the spike probability alternative is an actual effect and not Monte Carlo error.

Finally we investigate how the estimation error depends on the choice of $\lambda$. We evaluate both the jump estimator $\tilde{p}_n$ and the estimator using a constant $\lambda$, i.e. the original SLWE $\hat{p}_n^{\lambda}$. The results are shown in Figure \ref{fig:7}.
\begin{figure}
  \centering
  \begin{tabular}{cc}
  \includegraphics[width = 0.5\textwidth]{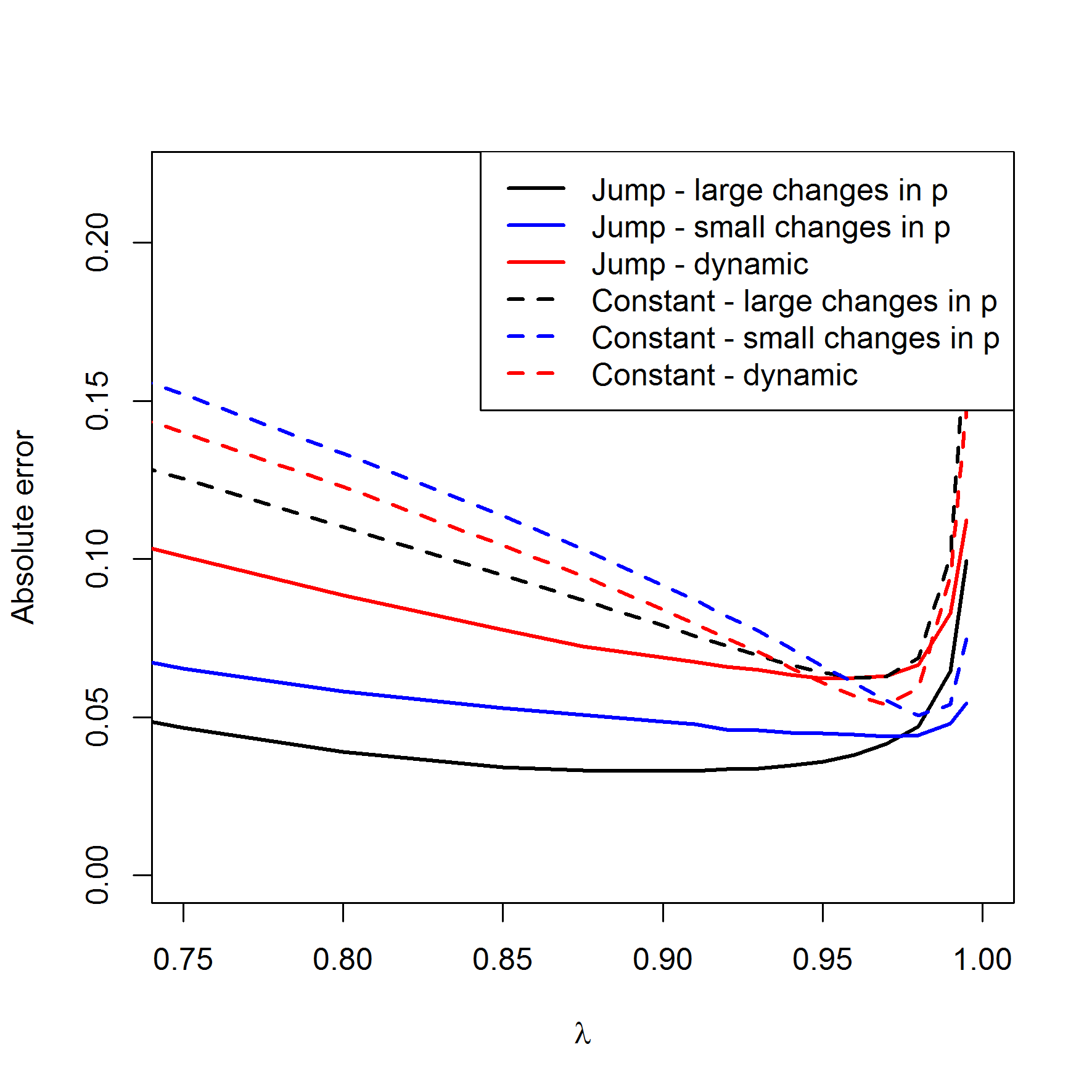} & \includegraphics[width = 0.5\textwidth]{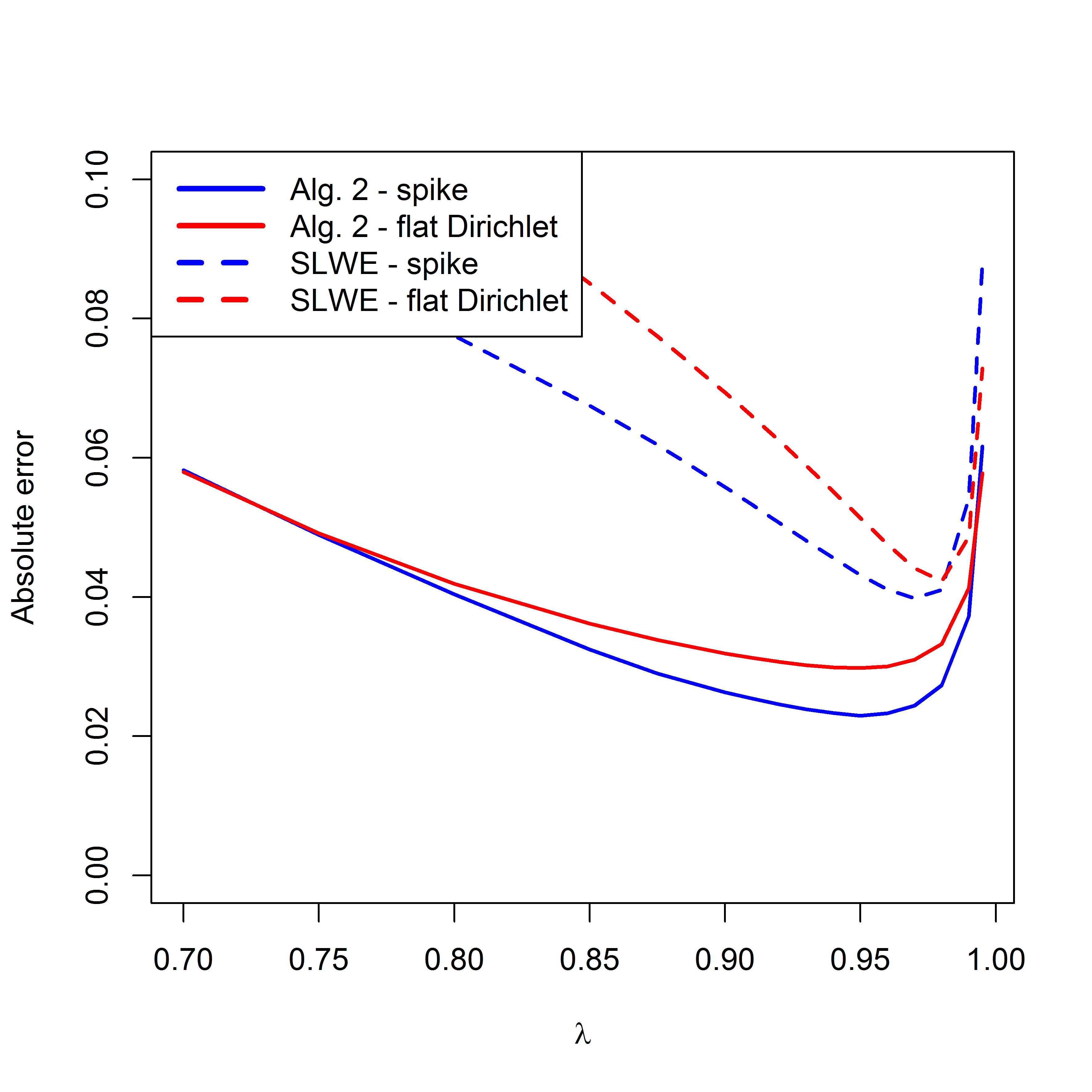}
  \end{tabular}
  \caption{Estimation error as a function of $\lambda$. The solid and dashed lines refer to the estimators $\tilde{p}_n$ and $\hat{p}_n^{\lambda}$, respectively. For the left panel, the black, blue and red curves refer to experiments where the changes in $p$ are large, small and dynamic. For the right panel, the blue and red curves refer to experiments where the changes in $p$ are based on spike probability and flat Dirichlet alternatives, respectively.}
  \label{fig:7}
\end{figure}
We start by discussing the binomial case (left panel). Comparing the solid and dashed black curves we see that $\tilde{p}_n$ outperforms  $\hat{p}_n^{\lambda}$ with a large margin. We also observe that for $\hat{p}_n^{\lambda}$ an optimal value of $\lambda$ is about 0.96. We also observe that the optimal value for $\lambda$ for the jump estimator is about 0.9. Recall that this is the $\lambda$ we should choose for the weak estimator with constant $\lambda$ that runs in parallel with the sample mean. This difference may come as a surprise, but remember that the purpose of the weak estimator with constant $\lambda$ is different in these two cases. For $\hat{p}_n^{\lambda}$ (original SLWE) we chose $\lambda$ to minimize the estimation error. For the jump estimator, we chose $\lambda$ to detect changes in $p$ as fast as possible to rapidly perform a jump. When the changes in $p$ are small, we also observe that $\tilde{p}_n$ outperforms $\hat{p}_n^{\lambda}$ for all choices of $\lambda$ (blue curves). When $p$ is changing dynamically (red curves), the picture is, as expected, less clear and which estimator that performs the best depends on the choice of $\lambda$.

For the multinomial case (right panel), we see that the jump algorithm (Algorithm \ref{alg:2}) outperforms the multinomial SLWE with a large margin for both the spike probability and the flat Dirichlet alternatives.

From both panels in Figure \ref{fig:7}, we see that for all the cases the performance of the jump estimator $\tilde{p}_n$ is less sensitive on the choice of $\lambda$ compared to the SLWE with constant $\lambda$. Said in another way, the jump estimator, $\tilde{p}_n$, performs well for a large range of different choices of $\lambda$ while for the SLWE with constant $\lambda$, $\hat{p}_n^{\lambda}$, performs well only for a small interval of values for $\lambda$. This is a very nice property of the jump estimator since in practical situations we do not know what is an optimal value for $\lambda$.

\subsection{Real-life data example}
In this section, we investigate the problem of tracking topics or sentiment in online streams of text. Examples of such text streams could be online discussion threads and news/social media feeds like Twitter. A popular approach is to use keyword lists like sentiment lexicons. A keyword list is a set of words for each topic or sentiment type (for example: happy, sad, angry, etc). Such an approach is usually more robust to domain changes than machine learning approaches \cite{liu2012sentiment} which makes the keyword approach ideal for online tracking of topics or sentiment in discussion threads and news/social media feeds.


In the experiment in this section we consider the problem of online tracking of the current topic in a news feed. We assumed four topics, namely news about the European Union (EU), news about economy, sports and entertainment. We collected a large set of news articles about the four topics from the popular Norwegian online news paper site \texttt{vg.no}. We assumed that the instants when the text changed between the different topics were unknown to our algorithm. The task was to track the probabilities that the current topic is EU, economy, sports or entertainment.

We now want to apply the algorithms described in this paper for the topic tracking problem. We started by generating a keyword list for each of the four topics. The keyword list for a given topic were generated by choosing words that had a high Pointwise mutual information to the given topic \cite{manning1999foundations}.
We assumed that we received one word at the time from the news feed and every time we received a new word, we updated our probability estimates that the current topic were EU, economy, sports or entertainment. If the current word received from the news feed was part of the EU keyword list, we can think of this as an outcome '1' from a multinomial distribution. If the word was part of the economy keyword list, we can think of this as an outcome '2' from a multinomial distribution and so on. Using the weak estimator in equation \eqref{eq:7}, we can now update our estimate of the probability vector, namely the probabilities that the current topic is EU, economy, sports or entertainment. Similarly we can update the estimate of the probability vector using the jump algorithm in Algorithm \ref{alg:2}. All words that were not part of any of the keyword lists were removed from the text corpus.

A natural offline way to estimate of the probability that the current topic was EU (economy, sports, entertainment) based on the keyword lists was to compute the portion of all the keywords in an article that were EU (economy, sports, entertainment) keywords. We denote this the offline approach and can be seen as the optimal estimates for the probability of the different topics based on the keyword lists. In an online setting it is not possible to compute the offline estimates, but ideally we want the online estimators in \eqref{eq:7} and in Algorithm \ref{alg:2} to be as close as possible to the optimal offline estimates. We compare the performance of the online estimators in this paper by measuring how close they were to the optimal offline approach. When performing the experiments we ran a two fold cross validation where we used half of the articles to compute the keyword lists and the other half to track the probabilities that the current topic was EU, economy, sports or entertainment. Next, we switched and trained and tested in the opposite direction.


Figure \ref{fig:8} shows the tracking of the probabilities for the different topics
\begin{figure}
  \centering
  \includegraphics[width = \textwidth]{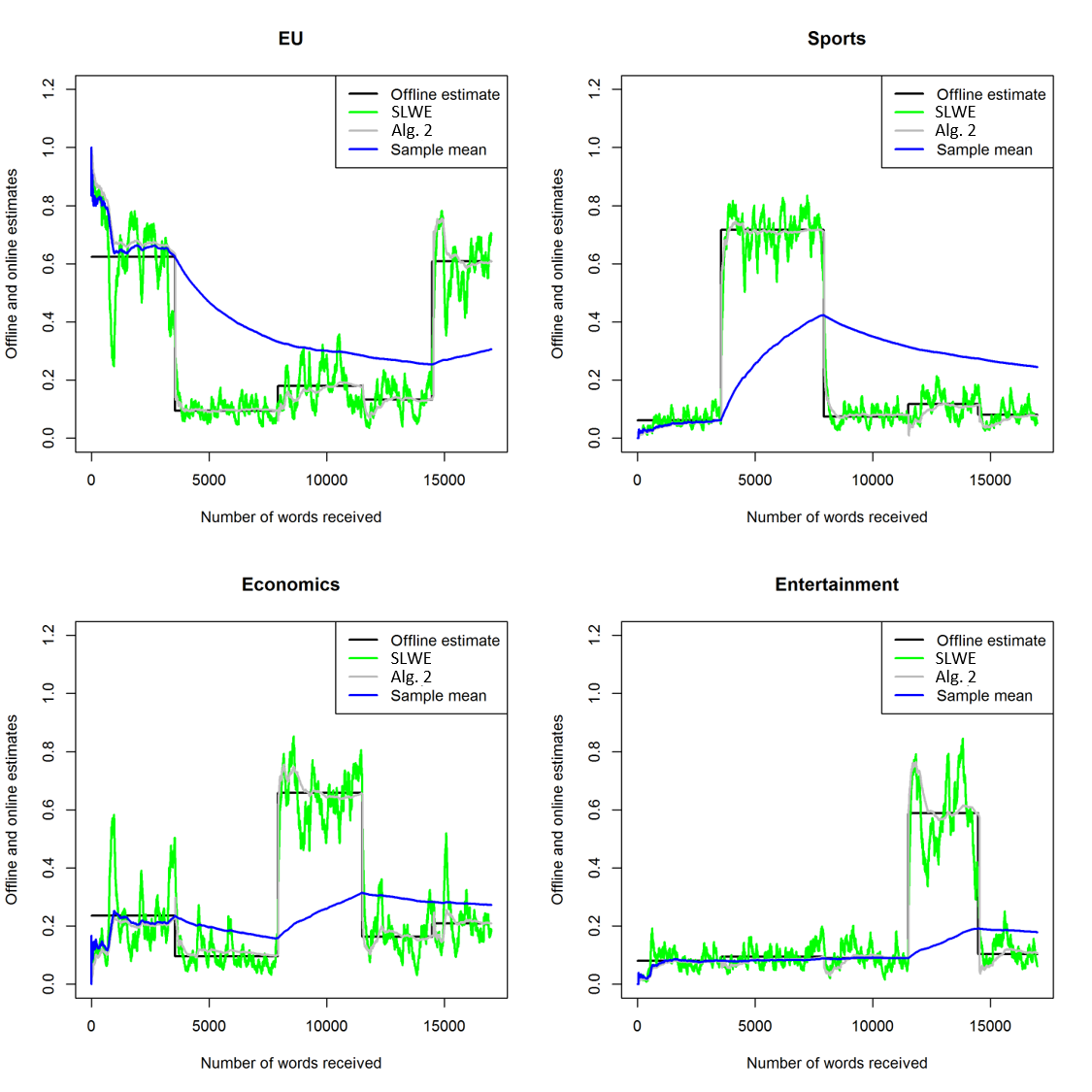}
  \caption{The panels from upper left to bottom right show the tracking of the probabilities that the current topic is EU, economy, sports or entertainment, respectively. The black curves shows the offline estimate every time a new word is received. The gray, green and blue curves show the jump estimator, the SLWE estimator and the sample mean, respectively.}
  \label{fig:8}
\end{figure}
We see that the jump estimator in Algorithm \ref{alg:2} adapts faster when the text stream changes topic and also tracks the offline estimates more efficiently in the stationary parts than the SLWE with constant $\lambda$. The jump estimator performed well for a large range of values for $\alpha$ and $\lambda$ but the best results were achieved using $\chi^2_{3,\alpha} \approx 30$ and $\lambda = 0.96$. The SLWE estimator, $\hat{p}_n^{\lambda}$, performed the best using $\lambda = 0.99$. The mean absolute estimation error compared to the offline estimator where 0.0235 and 0.0505 for the jump and the SLWE estimators, respectively, which means that the jump estimator clearly outperforms the SLWE estimator for this application.

\section{Closing remarks}
\label{sec:closrem}

In this paper we have constructed an estimation procedure that combines the strengths of a weak estimator with constant $\lambda$ and and decreasing $\lambda$ (sample mean). We have developed a hypothesis test procedure to rapidly detect a change in the underlying $p$. Further we have proposed an efficient procedure to jump to a new estimate when a change is detected. The experiments show that the procedure efficiently detects changes in the underlying distribution and outperforms the original SLWE with constant $\lambda$ with a large margin.

The experiments also showed that the performance of the jump estimator $\tilde{p}_n$ is less sensitive to the choice of $\lambda$ compared to the SLWE with constant $\lambda$. Said in another way, the jump estimator performs well for a large range of different choices of $\lambda$ while the SLWE with constant $\lambda$, $\hat{p}_n^{\lambda}$, performed well only for a small range of choices for $\lambda$. This is a very attractive property of the jump estimator since in practical situations we do not know what is an optimal value for $\lambda$.

\bibliographystyle{plain}
\bibliography{ParallelEstimBib}

\end{document}